\documentclass[
notitlepage,
floatfix,
aps,
pra,
reprint,
twocolumn,
superscriptaddress,
nofootinbib]{revtex4-2}
\usepackage[T1]{fontenc}
\usepackage{times}
\usepackage{graphicx}
\usepackage{amsmath,amssymb,amsthm}
\usepackage{enumitem}       
\usepackage[linesnumbered,ruled,vlined]{algorithm2e}      
\usepackage{subcaption}     

\def\openone{\leavevmode\hbox{\small1\kern-3.8pt\normalsize1}}

\def\II{\mathbb I}

\def\cg{{\cal G}}

\def\RR{\mathbb{R}}

\def\11{\mathbb{I}}

\newtheorem{definition}{Definition}[section]

\newtheorem{lemma}[definition]{Lemma}

\newtheorem{theorem}[definition]{Theorem}

\newcommand{\outerp}[2]{\ket{#1}\!\bra{#2}}

\newcommand\norm[1]{\lVert#1\rVert_1}
\newcommand\normlarge[1]{\left\lVert #1 \right\rVert_1}
\newcommand\pnorm[2]{\lVert#1\rVert_#2}
\newcommand\pnormp[2]{\left\lVert#1\right\rVert_#2^#2}

\newcommand{\supp}{\mathop{\rm supp}\nolimits}

\usepackage{times}

\newcommand{\bra}[1]{\langle#1|}
\newcommand{\ket}[1]{|#1\rangle}

\newcommand{\cD}{{\cal D}}

\newcommand{\cG}{{\cal G}}

\newcommand{\cH}{{\cal H}}
\newcommand{\cK}{{\cal K}}

\newcommand{\cP}{\mathcal{P}}

\newcommand{\cZ}{{\cal Z}}

\DeclareMathOperator{\Tr}{Tr}

\newcommand{\Renyi}{R{\'e}nyi~}

\def\accept{{\Omega_{\textnormal{acc}}}}
\def\AT{{\Omega_{\textnormal{AT}}}}
\def\EV{{\Omega_{\textnormal{EV}}}}

\usepackage{hyperref}

\usepackage[titletoc]{appendix}

\usepackage[dvipsnames]{xcolor}

\begin{document}

\title{Generalized Numerical Framework for Improved Finite-Sized Key Rates with \Renyi Entropy}

\author{Rebecca R.B. Chung}
\affiliation{School of Physical and Mathematical Sciences, Nanyang Technological University, Singapore 637371}
\author{Nelly H.Y. Ng}%
\email{nelly.ng@ntu.edu.sg}
\affiliation{%
 School of Physical and Mathematical Sciences, Nanyang Technological University, Singapore 637371}%
\affiliation{Centre for Quantum Technologies, National University of Singapore, 3 Science Drive 2, 117543, Singapore}
\author{Yu Cai}
\affiliation{%
School of Physical and Mathematical Sciences, Nanyang Technological University, Singapore 637371} %

\date{\today}
\begin{abstract}
Quantum key distribution requires tight and reliable bounds on the secret key rate to ensure robust security. This is particularly so for the regime of finite block sizes, where the optimization of generalized \Renyi entropic quantities is known to provide tighter bounds on the key rate. However, such an optimization is often non-trivial.
In this work, we present an analytical bound on the \Renyi entropy in terms of the \Renyi divergence and derive the analytical gradient of the \Renyi divergence. This enables us to generalize existing state-of-the-art numerical frameworks for the optimization of key rates. With this generalized framework, we show improvements in regimes of high loss and low block sizes, which are particularly relevant for long-distance satellite-based protocols. 
\end{abstract}

\maketitle

\section{\label{sec:Introduction} Introduction}

Quantum key distribution (QKD) is a cryptographic method for establishing secret keys between two distant parties, commonly called Alice and Bob, by exchanging quantum signals through a public quantum channel and using classical communication over an authenticated public channel~\cite{BENNETT20147,ekert1991quantum}. An adversary, referred to as Eve, can fully access the classical communication and potentially interfere with the quantum channel. Since the introduction of the first QKD protocol~\cite{BENNETT20147}, QKD has evolved into various forms~\cite{wolf2021quantum,diamanti2016practical, xu2020secure, pirandola2020advances} to enhance performance, simplify implementation, and address vulnerabilities. 

A key performance metric for a QKD protocol is the secret key rate, defined as the ratio between the number of secret key bits and the number of raw quantum bits exchanged during the protocol. The secret key rate is determined by the full set of parameters governing the protocol such as noise and loss of the quantum channel, and assumptions on the attacking power of Eve. A central ingredient in formulating the key rate is the leftover hashing lemma~\cite{renner2008security, tomamichel2011leftover} which quantifies the number of secret bits distillable from privacy amplification requiring Eve has no knowledge of the final key. Under the assumption of independently and identically distributed (i.i.d.) signals, one would be able to obtain a bound on the smooth min-entropy via the von Neumann entropy \cite{tomamichel2012framework}. The bound saturates in the asymptotic limit where the number of signals are infinite, due to the asymptotic equipartition property~\cite{tomamichel2012framework}. Early works ~\cite{devetak2005distillation, berta2010uncertainty} subsequently formally established the security for collective attacks and key rate formulas in the asymptotic regime. 
The security against general attacks are also established via the entropy accumulation theorem (EAT) \cite{dupuis2020entropy,george2022finite,metger2023security,metger2024generalised} and the postselection technique \cite{christandl2009postselection}. 

Since imperfections in sources and detectors in a practical QKD setup are inevitable, it quickly became evident that robust security proofs should be derived assuming that Eve may fully exploit such imperfections to her advantage~\cite{gottesman2004security,scarani2009practical}. In essence, a practical QKD security proof necessitates further numerical optimizations based on the experimental parameters estimated by the honest parties. This optimization considers the worst-case scenario where Eve exploits device imperfections to perform an optimal attack. Winick et al. \cite{winick2018reliable} developed a reliable numerical framework to evaluate the von Neumann key rate under collective attack through a two-step approach involving semi-definite programming (SDP). George et al. \cite{george2021numerical} subsequently extended the framework by incorporating finite size analysis, rendering it a highly valuable tool for evaluating the key rates of QKD implementations. 

In a recent work by Dupuis ~\cite{dupuis2023privacy}, it was shown that the key rate can be formulated directly in terms of the generalized \Renyi entropies via a \Renyi leftover hashing lemma. This new lemma eliminates the need for smoothing which is a challenging optimization problem over quantum states, and avoids loose bounds associated with the von Neumann entropy in the finite-size regime. By optimizing over a single \Renyi parameter $\alpha$, this approach is expected to yield tighter bounds in the finite-size regime. Recent works~\cite{tupkary2024security,arqand2024generalized, kamin2024finite,nahar2024postselection} have incorporated the \Renyi leftover hashing lemma into key rate analysis. 
However, the bounds on the \Renyi key rate remain loose, and numerical frameworks for the \Renyi entropy key rate evaluations remain undeveloped due to the inherent challenges in optimizing the \Renyi entropy. 

In this work, we address these challenges by generalizing the numerical framework initially proposed by Winick et al. \cite{winick2018reliable} to include the optimization of the generalized \Renyi key rate. We derive a bound to the \Renyi entropy in terms of the \Renyi divergence. Such bound is relatively tighter compared to older works \cite{renner2008security,scarani2008finite,scarani2009practical,george2021numerical} where the von Neumann entropy is used to analyze finite-size key rates. We subsequently obtain an analytical expression for the gradient of the \Renyi divergence. These inputs are then utilized in the Frank-Wolfe algorithm \cite{frank1956algorithm} as part of the minimization algorithm. Additionally, we incorporate the finite-size framework originally introduced by George et al. \cite{george2021numerical} into our analysis for collective attacks and using the open-source numerical package (openQKDsecurity)~\cite{burniston_2024_14262569}. Consequently, we achieve an improvement in key rate, which we demonstrate to be pertinent in practical scenarios, including in the presence of channel loss and noise.

This manuscript is organized as such: Section \ref{sec: Background} provides the background of the basic QKD protocol framework, and 
in Section \ref{sec: technical formulation} we introduce our technical results and the formulation of the generalized framework. In Section \ref{sec: numerical analysis} we conduct a numerical analysis using our generalized framework. Section \ref{sec: discussion} concludes with a discussion and outlook for future work.

\section{\label{sec: Background} QKD Protocol Framework}
We begin by providing a brief overview of the basic QKD protocol framework and its mathematical formulations used for the security analysis. Below is a summary of the steps for a generic protocol.
\begin{enumerate}[leftmargin=*]
\item For each round of the protocol $i\in\{0,1,\dots N\}$, 
\begin{enumerate}[topsep=3pt,itemsep=0ex,leftmargin=15pt]
    \item \textbf{State preparation and transmission:} 
    In the entanglement-based (EB) setting, a source prepares an entangled state and sends it to two trusted parties Alice and Bob via a public quantum channel. In the prepare-and-measure (PM) setting, the source is located in Alice's lab and the state is modeled by the source-replacement scheme \cite{bennett1992quantum, curty2004entanglement,ferenczi2012symmetries} where Alice prepares a bipartite entangled state $\ket \Phi_{AA'} =\sum_x \sqrt{p_x} \ket x_A \otimes \ket {\varphi_x}_{A'}$, keeps $A$ and sends $A'$ to Bob. The system $A'$ goes through a quantum channel $A'\rightarrow B$ before reaching Bob. 
    \item \textbf{State measurement:} Alice measures $A$ using a set of POVM and stores her measurement result in a classical register $X_A$. Bob measures $B$ using another set of POVM and stores his measurement result in a classical register $Y_B$. They subsequently store their choice of POVM in a public register.
    \item \textbf{Public announcements and sifting:} Alice and Bob exchange classical information from their public registers. The classical communication script is stored in a classical register $S$ for sifting. During this step, Eve also obtains a copy of $S$ and stores at $E_S$. Alice and Bob subsequently either keep or discard their measurement results depending on $S$. 
\end{enumerate}
\item \textbf{Parameter estimation}: Alice and Bob randomly choose a small proportion of signals to construct a frequency distribution, which is used to estimate errors on the quantum channel, while the remainder of the signals are used for key distillation. The protocol is accepted if the frequency distribution falls within a set of previously agreed statistics, and is aborted otherwise; we denote this event as $\Omega_{\rm AT}$. The remaining number of signals after parameter estimation is denoted as $n$, which is the key distillation rounds.
\item \textbf{Key Map: } Alice subsequently applies a key map procedure on $A$ according to $X_A$ and maps it into an intermediary key register $R$. (In the reverse-reconciliation scheme, Bob applies key map instead.) She subsequently applies an isometry on $R$ and stores it in the final raw key register $Z_A$.
\item \textbf{Error correction and verification:} Alice and Bob exchange $\lambda_{EC}$ bits of classical information so that Bob can perform error correction on his string and produce a guess of Alice's raw key. To verify the correctness of this step, Alice sends a 2-universal hash of her key string of length $\log_2(1/\epsilon_\textnormal{EV})$ to Bob for error verification. Bob subsequently compares Alice's hash with the hash of his guess. The protocol is accepted if the hash matches; we denote this event as $\Omega_{\rm EV}$, and $\epsilon_\textnormal{EV}$ as the failure probability of error verification.
\item \textbf{Privacy amplification:} Finally, Alice and Bob randomly pick a 2-universal hash function and apply it to their raw keys, producing the final secret keys $K_A$ and $K_B$ respectively. $\epsilon_\textnormal{PA}$ is denoted as the failure probability of privacy amplification.
\end{enumerate}

Upon completion of the protocol, $\Pr[\accept] = \Pr[\EV \wedge \AT] $ is the probability where the protocol is accepted ($\Pr[\bot]$ when aborted); in other words, $\accept$ refers to the events where both parameter estimation and error verification succeed. 
With this, we formally establish the security of the protocol via the following statement. 

\begin{definition}\label{def: Security of QKD} A key distillation protocol with $n$ key distillation rounds is said to be $\epsilon_{\rm sec}$-secret if
    \begin{equation}
        \frac{1}{2}\Pr[\accept]\normlarge{\rho_{K_A\tilde CE_n'|\accept} -\bar{\rho}_{K_A}\otimes\rho_{\tilde C E_n'|\accept}}\leq \epsilon_\textnormal{sec}.
    \end{equation}
    where $\tilde C$ is a classical register storing the error correction and verification scripts, and $E_n'=E_nE_{S^n}$ is Eve's total system containing her quantum side information $E_n$ and a copy of the public information $E_{S^n}$ over all the rounds, and $\bar \rho_{K_A}$ is a maximally-mixed state.\\
    Additionally, the protocol is $\epsilon_{\rm corr}$-correct if
    \begin{equation}
        \Pr[K_A \neq K_B \wedge \accept] \leq \epsilon_\textnormal{corr}.
    \end{equation}
Such a protocol is $\epsilon$-secure when $\epsilon_\textnormal{sec}+\epsilon_\textnormal{corr}=\epsilon$.
\end{definition}

It is worth noting that for general attacks by Eve over all $n$ rounds, we need to quantify security by assuming that Eve has access to the full system of $E_n'=E_nE_{S^n}$. However, when considering i.i.d. collective attacks, it suffices to only consider attacks on a single-round, where Eve's single-round system is simply denoted as $E'=EE_S$.

\section{\label{sec: technical formulation} Numerical Framework for Key Length Optimization}
\subsection{Key Length formulation \label{subsec: key length formulation}}
We first formulate the secret key length using the recently established leftover hashing lemma for \Renyi entropies \cite{dupuis2023privacy}, which quantifies the amount of secret bits obtained from the 2-universal hash function during privacy amplification. 
\begin{theorem}\label{thm: PA without smoothing}
\textnormal{Privacy Amplification (\cite{dupuis2023privacy}, Theorem 8)}. For $\alpha \in (1,2]$, we have $F:Z_A\rightarrow K_A$ as a family of two-universal hash function with $K_A=\{0,1\}^l$. If the hash function is drawn uniformly from $F$, then
\begin{equation}
\begin{aligned}
    &\normlarge{\rho_{K_A\tilde C E_n'|\accept}-\frac{\II_{K_A}}{|\cK_A|}\otimes\rho_{\tilde CE_n'|\accept}} \leq 2^{g(l,\alpha,\rho)},\\
    &g(l,\alpha,\rho) = \frac{2}{\alpha}-1+\frac{1-\alpha}{\alpha}\left[ H_\alpha(Z_A^n|\tilde{C}E_n')_{\rho|\accept} - l\right],
\end{aligned}
\end{equation}
where $l=\log_2(|\cK_A|)$ is the final length of the secret keys after hashing and $H_\alpha$ denotes the conditional Rényi sandwiched entropy.
\end{theorem}

A central technical quantity in Theorem \ref{thm: PA without smoothing} is the conditional \Renyi sandwiched entropy, which is defined as a special case of the sandwiched \Renyi divergence. 
\begin{definition}\label{def:entropy}
    (Sandwiched \Renyi divergence and conditional entropy,  \cite{tomamichel2015quantum}) For subnormalized quantum states $\rho,\sigma $ where $ \supp(\rho) \subseteq \supp(\sigma)$ and real values of $\alpha \in (0,\infty)$, the sandwiched \Renyi divergence as
    \begin{equation}
        D_\alpha(\rho||\sigma) := 
            \frac{1}{\alpha-1}\log_2\frac{\pnormp{\sigma^{\frac{1-\alpha}{2\alpha}}\rho\sigma^{\frac{1-\alpha}{2\alpha}}}{\alpha}}{\Tr(\rho)} ,  
    \end{equation}
where $\pnorm{\rho}{p}=[\Tr(\rho^p)]^{1/p}$ is the Schatten p-norm.
For states $\rho_{AB}$, the corresponding conditional \Renyi entropy is subsequently defined as
\begin{align}
H_\alpha (A|B)_\rho &:= -\inf_{\sigma_B}D_\alpha(\rho_{AB}||\II_A \otimes \sigma_B). 
\end{align}
\label{eq: conditional sandwiched renyi with optimization}
\end{definition}
\noindent The conditional \Renyi entropy is monotonically non-increasing in $\alpha$, and satisfies data processing inequality \cite{beigi2013sandwiched}. It also reduces to the von Neumann entropy at $\alpha\rightarrow 1$, and the conditional min-entropy when $\alpha \rightarrow \infty$. It thus recovers properly the conventional quantum entropic measures in a well-behaved manner, giving us a valuable information-theoretic measure to capture non i.i.d. features of general quantum states. 

Following from Lemma 9 of Ref \cite{tupkary2024security}, the \Renyi entropy can then be lower-bounded as shown 
\begin{equation}
H_\alpha(Z^n_A|\tilde{C}E_n')_{\rho|\accept} \geq H_\alpha (Z_A^n| \tilde{C}E_n')_{\rho} - f(\alpha)   
\label{eq: Decrease accept}
\end{equation}
with a correction term $f(\alpha) = \frac{\alpha}{1-\alpha}\log_2\left(\Pr[\accept]\right)$. Here, we let $\tilde C = CV$, where $C$ is the register containing error correction scripts and $V$ is the register containing error verification scripts. Subsequently, the \Renyi entropy is further bounded by \cite{tomamichel2015quantum} 
\begin{equation}
    \begin{aligned}
        H_\alpha &(Z_A^n| \tilde{C}E_n')_{\rho}\geq H_\alpha (Z_A^n| CE_n')_{\rho}-\log_2\left(\frac{1}{\epsilon_{\textnormal{EV}}}\right).
    \end{aligned}
    \label{eq: Decrease EV}
\end{equation}
Finally, to account for key bits used for error correction, the \Renyi entropy decreases due to information leakage, 
    \begin{equation}
        H_\alpha(Z_A^n|CE_n')_\rho \geq H_\alpha(Z_A^n|E_n')_\rho - \lambda_{EC},
        \label{eq: Decrease EC}
    \end{equation}
where $\lambda_{EC}=nf_{EC}H(Z_A|Y_B)$ is the amount of bits leaked during public communication for error correction, and $f_{EC}$ is the error correction efficiency.
 
Now we are ready to formulate the objective function for our optimization problem, where the goal is to optimize $H_\alpha(Z_A^n|E_n')_\rho$. Starting from the initial state $\rho_{AB}$ shared between Alice and Bob, we first define a postprocessing map $\cG(\rho_{AB})=\sum_s G_s \rho_{AB} G_s^\dagger$ according to Ref \cite{lin2019asymptotic} that describes the state measurement, public announcements and key map procedures, where
\begin{equation}\label{eq:G_s}
    G_s = \sum_x \ket {g(x,s)}_R \otimes  M_A^{(x,s)} \otimes \II_B \otimes p_s\ket {s}_S.
\end{equation}
Here, $M_A^{(x,s)}$ is Alice's choice of POVM and $p_s$ is the probability of a particular choice of measurement $s$. The function $g(x,s)$ maps the values $x$ and $s$ into an intermediary key register $R$. We thus have $\cG(\rho_{AB}) = \rho_{RABS}$. 
We then proceed to formulate our objective function by first obtaining a bound to $H_\alpha(Z_A^n|E_n')_\rho$. Under the i.i.d. collective attack scenario, the \Renyi entropic bound for each round is the same, thus we can express the Rényi entropy in terms of the single-round bound as
\begin{equation}
\begin{aligned}
    H_\alpha(Z_A^n|E_n')_\rho = p_\textnormal{gen}*N*\Tr(\cG(\rho_{AB}))*H_\alpha(Z_A|E')_\rho,
\end{aligned}
\label{eq: iid renyi entropy}
\end{equation}
where $p_\textnormal{gen}$ is the proportion of signal rounds used for key generation during parameter estimation, $N$ is the total number of rounds of the protocol, and $\Tr(\cG(\rho_{AB}))$ is the probability that the signal is accepted during the sifting process for each round. Next, we formulate the single-round \Renyi entropy via the following theorem.

\begin{theorem}
\label{thm: duality renyi divergence}
    Let $\rho_{Z_ARABSE'}$ be pure. Then, for $\gamma = (2-\alpha^{-1})^{-1}$ and $\beta=\alpha^{-1}$, the Rényi entropy for a single round for $\alpha \in (1,2]$ can be expressed as
    \begin{align}
    H_\alpha(Z_A|E') &=\inf_{\sigma_{RABS}}D_\gamma(\cG(\rho_{AB})||\cZ(\sigma_{RABS})),\label{eq:D_gamma} \\
    &\geq D_\beta\bigl(\cG(\rho_{AB})||\cZ\circ\cG(\rho_{AB})\bigr) ,\label{eq: divergence bound}
    \end{align}
    Here, $\cZ(\cdot)=\sum_i(Z^i_R\otimes \II_{ABS})(\cdot)(Z^i_R\otimes \II_{ABS})$ is a pinching map with rank-one projectors $Z_R^i$ in the $Z$ basis acting on $R$. 
\end{theorem}

The proof for Theorem \ref{thm: duality renyi divergence} can be found in Appendix \ref{app: duality renyi divergence}. Theorem \ref{thm: duality renyi divergence} allows us to make use of the convexity of the sandwiched \Renyi divergence for $\beta \in (0.5,1)$ to optimize the single-round \Renyi entropy. From the above, Eq. ~\eqref{eq:D_gamma} gives a tight bound to the \Renyi entropy in the form of a \Renyi divergence. However, it is worth noting that it is helpful to further relax the bound to Eq. \eqref{eq: divergence bound} to remove the infimum, since the infimum demands an additional optimization procedure which cannot be trivially implemented by the Frank-Wolfe algorithm. Thus, using Eq. \eqref{eq: iid renyi entropy}, we now have 
\begin{equation}
\begin{aligned}
    H_\alpha(Z_A^n|E_n')_\rho \geq \min_{\rho_{AB}} f_\beta(\rho_{AB})*p_\textnormal{gen}*N,
\end{aligned}
\label{eq: Renyi LB objective}
\end{equation}
where the objective function to be optimized is
\begin{equation}
    f_\beta(\rho_{AB}) = \Tr(\cG(\rho_{AB}))\cdot D_\beta\Bigl(\cG(\rho_{AB})\bigl|\bigr|\cZ\circ\cG(\rho_{AB})\Bigr).
    \label{eq: objective function}
\end{equation}
Finally, by combining Eqs. \eqref{eq: Decrease accept}, \eqref{eq: Decrease EV} and \eqref{eq: Decrease EC} with \eqref{eq: Renyi LB objective}, and substituting into Theorem \eqref{thm: PA without smoothing}, we have the desired secrecy relation
\begin{equation}
\begin{aligned}
    &\frac{1}{2}\Pr[\accept]\normlarge{\rho_{K_A\tilde CE_n'|\accept} -\bar{\rho}_{K_A}\otimes\rho_{\tilde C E_n'|\accept}}\\
    &\leq 2^{-\frac{\alpha-1}{\alpha}\left(H_\alpha(Z_A^n|E_n')_\rho+ 2 + \log_2(\epsilon_\textnormal{EV})-\lambda_{EC}- l\right)} \\
    &\leq 2^{-\frac{\alpha-1}{\alpha}\left(\min_{\rho_{AB}}f_\beta(\rho_{AB})*p_\textnormal{gen}*N + 2 + \log_2(\epsilon_\textnormal{EV})-\lambda_{EC}- l\right)} \\
    &\leq \epsilon_\textnormal{PA}.
\end{aligned}
\label{eq: Final secrecy relation}
\end{equation}
The key length can then be written as 
\begin{equation}
\begin{aligned}
    &l_\textnormal{ i.i.d.}\leq \min_{\rho_{AB}} f_\beta(\rho_{AB})*p_\textnormal{gen}*N - g(\alpha),\\
    &g(\alpha) = \frac{\alpha}{\alpha-1}\log_2\left(\frac{1}{\epsilon_{\textnormal{PA}}}\right) +\lambda_{EC}+\log_2\left(\frac{1}{\epsilon_\textnormal{EV}}\right)-2. 
\end{aligned}
\label{eq: Final key length}
\end{equation} 

\subsection{Numerical Optimization of \Renyi Entropy}
Next, we will optimize the objective function derived in the previous section by generalizing the numerical framework in~\cite{winick2018reliable}, which obtains the lower bound to the key length via a two-step process. The original framework minimizes the von Neumann entropy to obtain the worst-case scenario of an eavesdropping attack. We will show how this framework can be extended to the optimization of \Renyi entropies. We begin by first defining the set of states over which the optimization takes place. Following from Ref \cite{george2021numerical}, the set of density matrices over which the optimization is performed for an EB protocol is defined as
\begin{equation}
\begin{aligned}
    \mathbf{S} = &\{ \rho \in \cD(\cH_A\otimes \cH_B)| \quad \exists F \in \cP(\Sigma) \\
    &\textnormal{such that}  
    \quad \norm{\Phi_\cP(\rho)- F} \leq\mu \\
    &\quad \quad \quad \quad \quad \norm{F - \bar{F}}\leq t\},
\end{aligned}
\label{eq: step 1 set}
\end{equation}
where $\cP(\Sigma)$ is the probability distribution of size $\Sigma$. The first inequality constraint models the parameter estimation step, where only the set of states that are a $\mu$ distance away from the agreed frequency $F$ is accepted due to finite-size effects. Here, we have 
\begin{equation}
    \mu = \sqrt{2}\sqrt{\frac{\ln(1/\epsilon_\textnormal{PE})+|\Sigma|\ln(m+1)}{m}},
\end{equation} where $\epsilon_\textnormal{PE}$ is the error probability of the parameter estimation step, and $m$ is the number of signals used for parameter estimation. The map $\Phi_\cP(\rho) = \sum_i \Tr(\rho \tilde{\Gamma}_i) \outerp i i$, where $\tilde{\Gamma}_i=M_A^i\otimes M_B^i$ maps $\rho$ into a probability of the measurement outcome by Alice and Bob. The second inequality takes into account slight deviations of the agreed statistics $F$ from the ideal statistics $\bar{F}$ in a practical QKD implementation. For PM protocols, an additional constraint is needed to account for Alice's certainty of her measurement outcome, i.e. $\Tr(\Gamma_i \rho) = \gamma_i , \quad \forall i \in |\Sigma|$, where $\Gamma_i=M_A^i\otimes \II_B$. See Appendix \ref{appendix:constraint} to express Eq. \eqref{eq: step 1 set} in SDP form.

Next, we outline the detailed minimization algorithm introduced by Ref \cite{winick2018reliable}. The minimization algorithm is done in two steps. In the first step, the key rate is cast as a convex objective function $f_\beta(\rho)$ that is a function of $\rho$ over set $\mathbf{S}$. For this step, the Frank-Wolfe algorithm \cite{frank1956algorithm} is used to obtain a state close to the optimal. The algorithm begins by starting from an initial point $\rho_0$ that is close to the identity, i.e. $\min_{\rho_0\in \mathbf{S}} \norm{\rho_0-\II}$. Subsequently, a step size $\Delta \rho$ is computed iteratively to give the next point. $\Delta \rho$ is cast as an SDP i.e. $\Delta \rho := \arg \min_{\Delta \rho} \Tr[(\Delta \rho)^T \nabla f_\beta(\rho_i)]$ where $\nabla f_\beta(\rho_i)$ is the gradient of the objective function at each iteration. The iteration is terminated once the state obtained is close to the optimal point (See Algorithm \ref{algo:min} and Figure \ref{fig:frankewolfe}). 
\begin{algorithm}
\caption{Step 1 Minimization \cite{winick2018reliable}}\label{algo:min}
Let $\epsilon > 0, \rho_0 \in \mathbf{S}$ and set $i=0$. \\
Compute $\Delta\rho:= \arg \min_{\Delta \rho}\Tr[(\Delta \rho)^T\nabla f_\beta(\rho_i)]$ s.t. $\Delta \rho+\rho_i\in \mathbf{S}$.\\
If $\Tr[(\Delta \rho)^T\nabla f_\beta(\rho_i)]<\epsilon$ then STOP. \\
Find $\lambda \in (0,1)$ that minimizes $f_\beta(\rho_i+\lambda\Delta \rho)$. \\
Set $\rho_{i+1}=\rho_i+\lambda\Delta \rho,i\leftarrow i+1$. Repeat 2-5.
\end{algorithm}

Upon completion of the first step of the minimization, a linearization is performed about the near-optimal state obtained from Step 1 such that the optimal point $f_\beta(\rho^*)$ is lower bounded by the following (Eqns. 77-79 of Ref \cite{winick2018reliable})
\begin{align}
    f_\beta(\rho^*)&\geq f_\beta(\hat{\rho})+\Tr((\rho^*-\hat{\rho})^T\nabla f_\beta(\hat{\rho}))\\
    &\geq f_\beta(\hat{\rho})-\Tr(\hat{\rho}^T\nabla f_\beta(\hat{\rho}))+\min_{\sigma \in \mathbf{S}}\Tr(\sigma^T \nabla f_\beta(\hat{\rho})). \label{eq: step 1 optimization}
\end{align} 
\begin{figure}[htb!]    \includegraphics[width=0.8\linewidth]{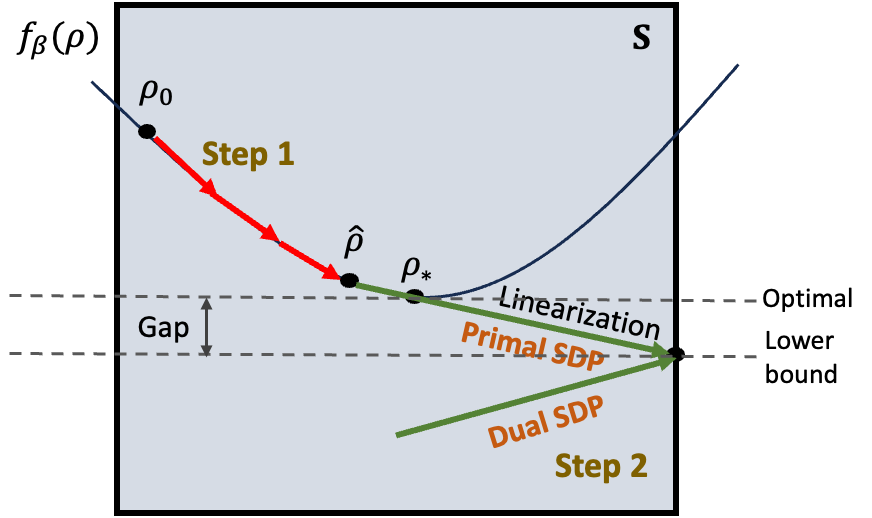}
    \caption{An illustration of the Frank-Wolfe minimization algorithm as outlined by Ref \cite{winick2018reliable}.}
    \label{fig:frankewolfe}
\end{figure}

Using the objective function $f_\beta(\rho)$ derived in Section \ref{subsec: key length formulation}, we can then apply the Frank-Wolfe algorithm to obtain the minimum key length. To do this, we also derive the gradient of the objective function which is needed for the Frank-Wolfe algorithm implementation and is shown in the Theorem below.
\begin{theorem}\label{thm: gradient func}
For $\beta \in [0.5,1)$, the gradient of the objective function $f_\beta(\rho)$ is given by
    \begin{equation}
    \begin{aligned}
        \nabla f_\beta(\rho) =&\cG^\dagger\biggl[\II \cdot D_\beta\Bigl(\cG(\rho)\bigl|\bigr|\cZ\circ\cG(\rho)\Bigr) \biggr] +\\
        &\biggl.\Tr\bigl(\cG(\rho)\bigr)\cdot \nabla D_\beta\Bigl(\cG(\rho)\bigl|\bigr|\cZ\circ\cG(\rho)\Bigr),
    \end{aligned}
    \end{equation}
    where the gradient of the \Renyi divergence is given by 
    \begin{equation}
    \begin{aligned}
        &\nabla D_\beta\Bigl(\cG(\rho)\bigl|\bigr|\cZ\circ\cG(\rho)\Bigr)\\ &= 
        \frac{1}{(\beta-1)}\cG^\dagger\left[\frac{\chi_1 + \chi_2 + \chi_3}{Q_\beta (\cG(\rho)||\cZ\circ\cG(\rho))} - \frac{\II}{\Tr(\cG(\rho))}\right],
    \end{aligned}
\end{equation}
and 
\begin{align}
    &\chi_1 = \cZ\left(\beta L(\mu)\int_0^\infty \frac{\cG(\rho)\cZ\circ\cG(\rho)^\mu \Xi^{\beta-1}}{(\cZ\circ\cG(\rho)+t)^{2}}t^\mu dt\right)\\
    &\chi_2 = \beta \cZ\circ\cG(\rho)^\mu \Xi^{\beta-1}\cZ\circ\cG(\rho)^\mu\\
    &\chi_3 = \cZ\left(\beta L(\mu)\int_0^\infty \frac{\Xi^{\beta-1}\cZ\circ\cG(\rho)^\mu\cG(\rho)}{(\cZ\circ\cG(\rho)+t)^{2}}t^\mu dt\right)\\
    &Q_\beta(\rho\|\sigma)=\pnormp{\Xi}{\beta}, \quad \Xi = \cZ\circ\cG(\rho)^{\mu}\cG(\rho)\cZ\circ\cG(\rho)^{\mu}, \\
    &\mu = \frac{1-\beta}{2\beta}, \quad L(\mu) = \frac{\sin(\pi\mu)}{\pi}.
\end{align} 
\end{theorem}
The proof of Thm. \ref{thm: gradient func} is given in Appendix \ref{app:gradproof}. 
Here, it is worth noting that 
\begin{align}
    &\int_0^\infty \frac{A}{(B+t)^2}t^\mu dt = \int_0^\infty (B+t)^{-1}A(B+t)^{-1}t^\mu dt\\
    &= \sum_{ij}\int_0^\infty (b_i+t)^{-1}\outerp{b_i}{b_i}A\outerp{b_j}{b_j}(b_j+t)^{-1}t^\mu dt \\
    &= L^{-1}(\mu)\sum_{ij}\frac{b_i^\mu-b_j^\mu}{b_i-b_j}\outerp{b_i}{b_i}A\outerp{b_j}{b_j},
\end{align}
where $\ket{b_i}$ is the eigenbasis of the matrix $B$ corresponding to eigenvalue $b_i$.
Finally, to ensure continuity of the gradient function, the gradient is perturbed with a depolarizing channel $\cD_\epsilon$, $    \cG_\epsilon(\rho) = \cD_\epsilon \circ\cG(\rho)$ for $0 < \epsilon < 1$ \cite{winick2018reliable}.

\section{\label{sec: numerical analysis} Numerical Analysis}
In this section, we conduct numerical analysis of the \Renyi key rate using our generalized framework. As an example, we consider the prepare-and-measure single-photon BB84 protocol under depolarizing and lossy channels. 

\begin{figure}[htb!]
\centering
\includegraphics[width=\linewidth]{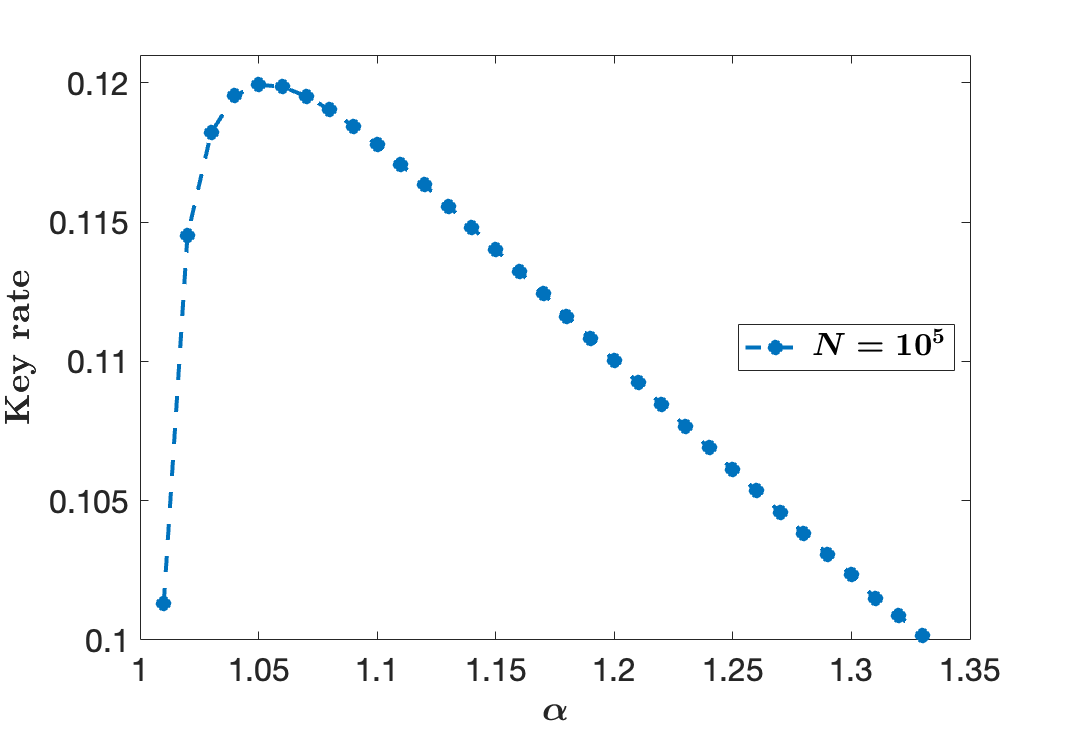}
    \caption{Numerical key rate at different values of $\alpha$ for $N=10^5$. }
    \label{fig:scan_alpha}
\end{figure}
\begin{figure}[htb!]
   \centering
\includegraphics[width=0.9\linewidth]{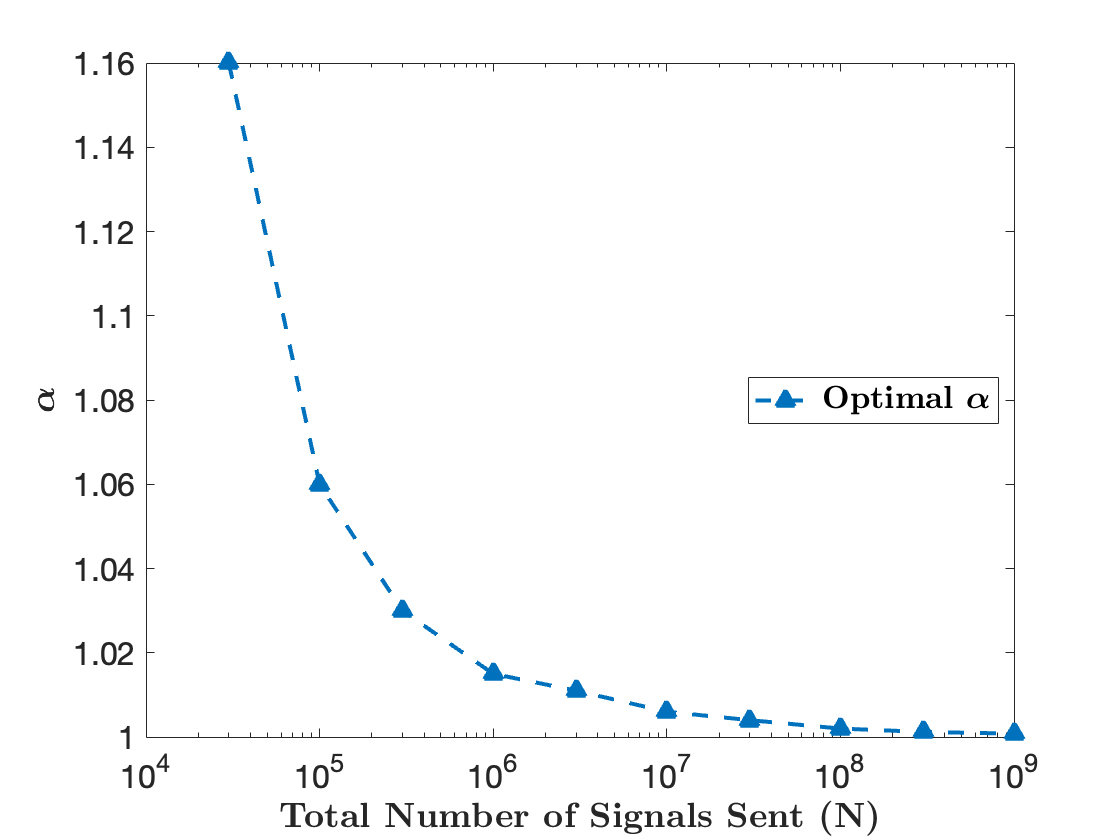}
    \caption{Optimal values of $\alpha$ for each signal block sizes $N$.}
    \label{fig: opt_alpha}
\end{figure}
We begin by studying how the key rate varies with $\alpha$ (See Figure \ref{fig:scan_alpha}). The $\alpha$ dependency of the key length comes from the \Renyi entropy term $H_\alpha (Z_A|E')$ as well as the factor of the $\log_2(1/(\epsilon_{PA}))$ term (See Eq.\eqref{eq: Final key length}). These two terms are competing effects since the \Renyi entropy is monotonically decreasing in $\alpha$ while the factor of the $\log_2$ term is increasing in $\alpha$. Hence, there exists an optimal $\alpha$ value for each signal block size $N$. 

The optimal $\alpha$ for each signal block size is then determined and plotted as shown in Figure \ref{fig: opt_alpha}. From the figure, it can be seen that the optimal $\alpha$ values converge to 1 as $N$ approaches infinity, due to asymptotic equipartition  property \cite{tomamichel2012framework,tomamichel2015quantum}. By optimizing $\alpha$ for each block sizes, we show the advantage of the \Renyi bounds obtained from our work compared to the original von Neumann bound proposed by Ref \cite{george2021numerical} in Figure \ref{fig:compare}. From Figure \ref{fig:compare}, it can be seen that the \Renyi key rate yields a higher bound as compared to the von Neumann key rate in the low block size regime ($N < 10^7$). This is especially evident for signal block size $N=10^5$, where the \Renyi key rate doubles that of the von Neumann key rate. For $N\geq 10^7$, there is no significant difference between the two bounds. 
\begin{figure}
\includegraphics[width=0.9\linewidth]{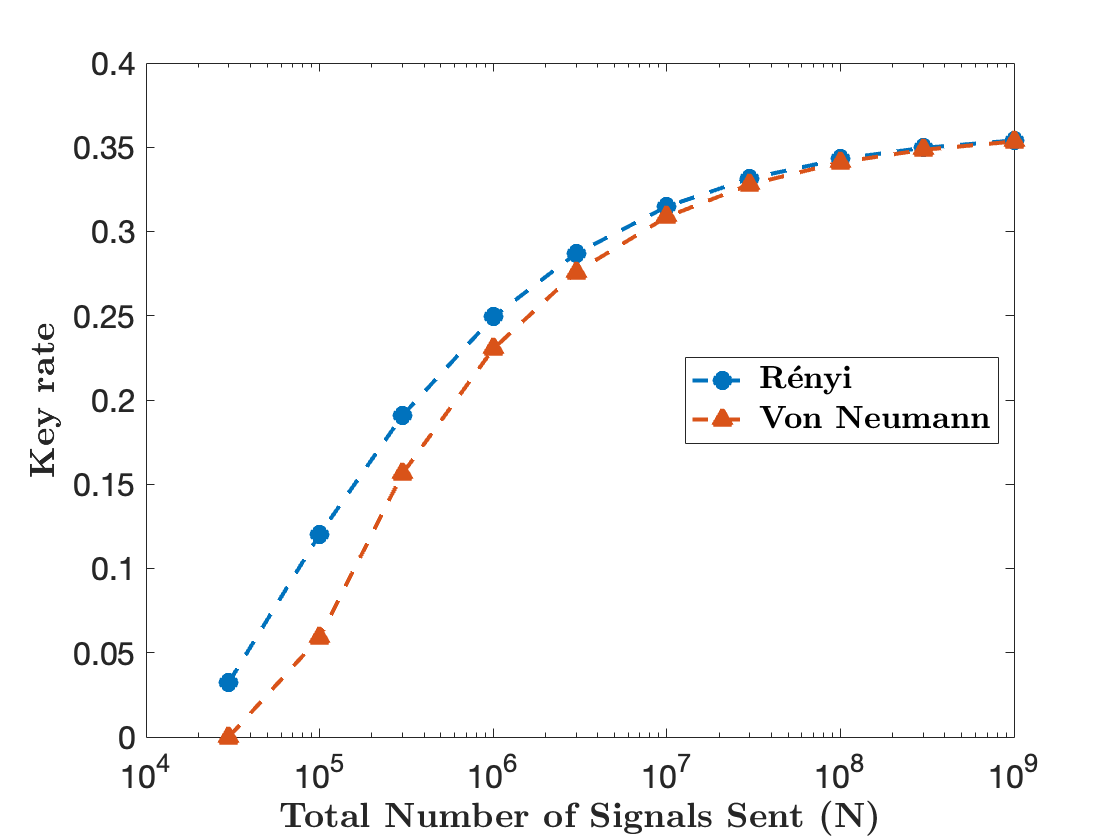}
    \caption{\Renyi key rate computed at optimal $\alpha$ vs the von Neumann key rate with depolarizing probability $p = 0.01$. }
    \label{fig:compare}
\end{figure}
\begin{figure*}[htb!]    
\includegraphics[width=0.8\linewidth,height=0.5\linewidth]{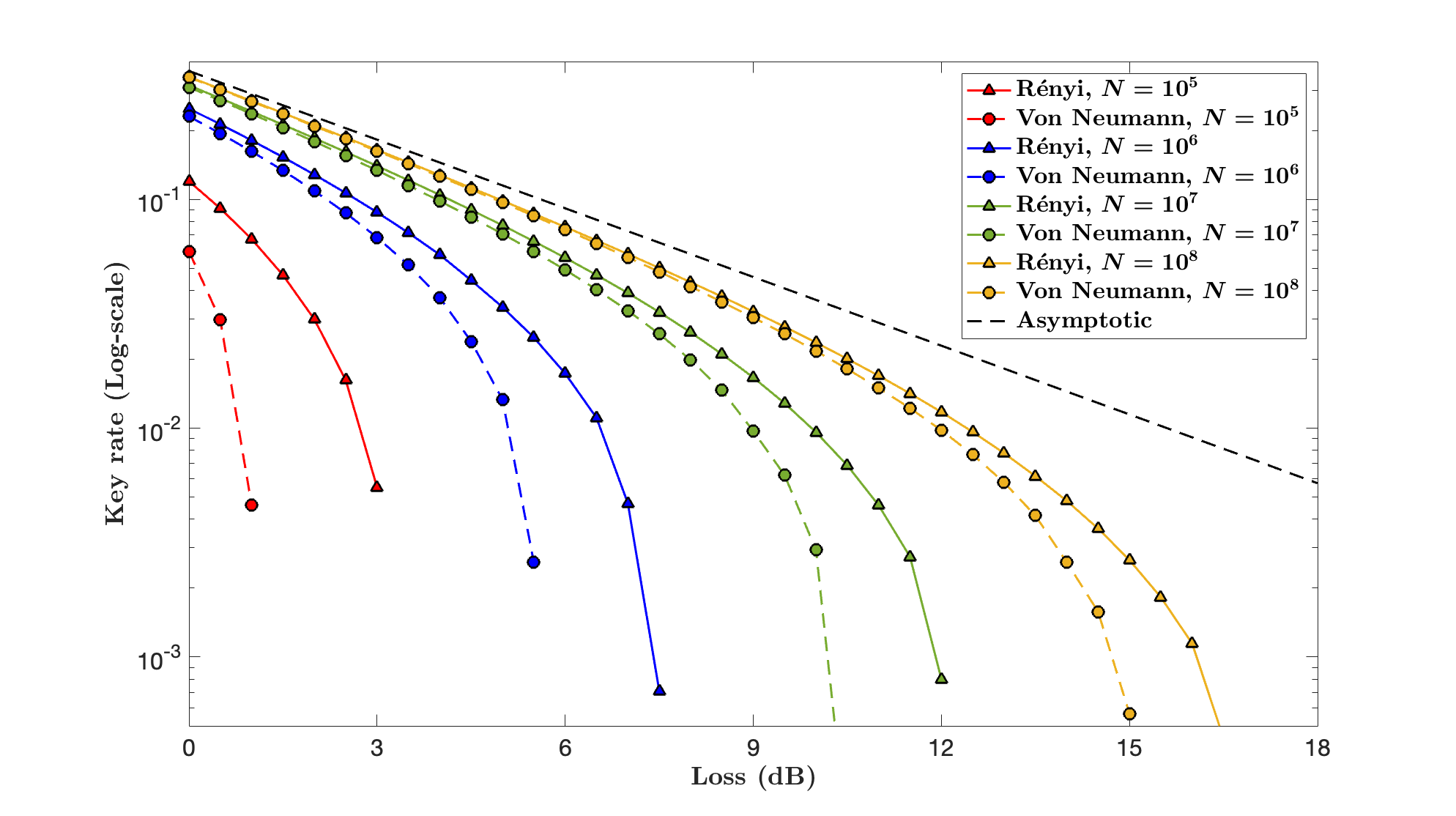}
    \caption{Key rate against channel loss with depolarizing probability $p = 0.01$.}
\label{fig:loss}
\end{figure*}
\par When comparing against different values of loss, the \Renyi entropy key rate also tends to outperform the von Neumann key rate in the high loss regime as seen in Figure \ref{fig:loss}. The \Renyi entropy key rate is shown to be more tolerant towards loss since the von Neumann key rate decreases to zero more rapidly in the high loss regime.

\section{\label{sec: discussion} Discussion and Outlook}
In this work, we have formulated the key rate in terms of the sandwiched \Renyi entropy under collective attack and derived a bound on the sandwiched \Renyi entropy using the \Renyi divergence, and subsequently generalizes the numerical framework originally proposed by Ref \cite{winick2018reliable}. An advantage of this generalization is the ability to optimize the \Renyi entropy in a straightforward way. By bounding the \Renyi entropy directly in terms of the \Renyi divergence via the duality relation, we are able to obtain a tighter lower-bound to the \Renyi key rate in the low block-sizes as well as high loss regime as compared to the conventional von Neumann bound. This improvement makes our generalized framework an important tool for the evaluation of long-range satellite-based QKD protocols, where signal block sizes remain small ($\sim 10^5$) with significant channel loss. 

Similar to the original paper by Winick et al., one drawback of our work is the need for perturbation of the gradient function to ensure a continuous gradient, which leads to numerical imprecision. One way to remedy this is to use the Gaussian interior point method instead \cite{hu2022robust,he2024exploiting}. A recent work \cite{he2025operator} has proposed an interior point solver that can be used to optimize sandwiched Rényi divergences. Future work can be done to integrate this solver into our work to optimize the Rényi entropic bound obtained in Theorem \ref{thm: duality renyi divergence}. In addition, we have also slightly relaxed the bound on the \Renyi entropy as seen in Eq. \eqref{eq: divergence bound}. There might be ways to further tighten this bound by introducing correction terms via continuity bounds on the \Renyi entropy, and we will leave this for future work. 
At the point of writing, we also became aware of a separate and independent work on \Renyi entropy key rate applied to decoy state QKD under coherent attacks \cite{kaminpreparation}. It would be interesting to see how our framework could be applied to the study of other protocols such as the decoy state protocol under the coherent attack assumption.

\section*{Acknowledgment}
The authors wish to thank Jing Yan Haw, Marco Tomamichel, Roberto Rubboli, Ray Ganardi, and Jakub Czartowski for helpful insights and discussions. The authors would also like to acknowledge Ernest Tan, Lars Kamin, and Devashish Tupkary for comments and discussions on the manuscript. This work is supported by the start-up grant of the Nanyang Assistant Professorship at the Nanyang Technological University in Singapore, MOE Tier 1 023816-00001 ``Catalyzing quantum security: bridging between theory and practice in quantum communication protocols", and the National Research Foundation, Singapore under its
Quantum Engineering Programme (National Quantum-Safe Network, NRF2021-QEP2-04-P01).

\section*{Code Availability}
The code used to obtain the data in this work can be found at \url{https://github.com/Rebeccacrb07/openQKDsecurity.git}. The data needed to verify or replicate the results of this
study are included in the article in a way that is interoperable
and reusable, e.g., tabular data.

\newpage
\nocite{*}

\providecommand{\noopsort}[1]{}\providecommand{\singleletter}[1]{#1}%

\clearpage
\onecolumngrid
\setcounter{page}{1}
\appendix
\begin{center}
    \textbf{\MakeUppercase{Supplementary Materials for ``A Generalized Numerical Framework for Improved Finite-Sized Key Rates with \Renyi Entropy''}}
\end{center}

\section{Proof of Theorem \ref{thm: duality renyi divergence}}
\label{app: duality renyi divergence}
Here we prove that the \Renyi entropy can be expressed in terms of the \Renyi divergence via the duality relation, where $\frac{1}{\alpha}+\frac{1}{\gamma}=2$ (Proposition 5.7 of Ref \cite{tomamichel2015quantum}) for $\alpha\in(1,2]$. 
During post-processing, we have that $\cG(\rho_{AB})=\rho_{RABS}$. Since Eve also obtains a copy of $S$ during the classical communication and stores it at $E_S$, the resulting state is $\rho_{RABSE_S}$, with its purification being $\rho_{RABSE_SE}=\rho_{RABSE'}$. Given the pure state $\rho_{RABSE'}$, Alice applies a key map isometry $V_Z=\sum_i \ket i_{Z_A}\otimes Z_R^i$ that models the $Z$ projective measurement on $R$ and identity on all other systems, and subsequently stores the result in the final key register $Z_A$ \cite{coles2012unification}.  We thus have $\rho_{Z_ARABSE'}=V_Z\rho_{RABSE'}V_Z^\dagger$ which is also a pure state. Then, 
\begin{align}
	H_\alpha (Z_A|E') &= -H_\gamma (Z_A|RABS) \\
	&= \inf_{\sigma_{RABS}}D_\gamma(\rho_{Z_ARABS}||\II_{Z_A}\otimes \sigma_{RABS})\\
	&\geq \inf_{\sigma_{RABS}}D_\gamma(V_Z^\dagger \rho_{Z_ARABS}V_Z || V_Z^\dagger (\II_{Z_A}\otimes \sigma_{RABS}) V_Z)\label{eq: DPI}\\
	&= \inf_{\sigma_{RABS}} D_\gamma(\rho_{RABS}||\sum_i Z_R^i(\sigma_{RABS})Z_R^i) \\
	&= \inf_{\sigma_{RABS}}D_\gamma(\cG(\rho_{AB})||\cZ(\sigma_{RABS})).
\end{align}

The first equality stems from the duality relation of the \Renyi entropy \cite{beigi2013sandwiched,muller2013quantum,zhu2017coherence} for pure state $\rho_{Z_ARABSE'}$. The inequality (\ref{eq: DPI}) is due to data-processing inequality, since $V_Z^\dagger$ is a partial isometry ($V_Z V_Z^\dagger \neq \II$).

\begin{align}
	H_\alpha (Z_A|E') &=-H_\gamma (Z_A|RABS)\\
	&\leq E_{R,\gamma}^{Z_A|RABS}(\rho_{Z_ARABS}) \\
	&= \inf_{\sigma\in\textnormal{sep}} D_\gamma(\rho_{Z_ARABS}||\sigma_{Z_ARABS}) \\
	&= \inf_{\sigma_{RABS}}D_\gamma(V_Z\rho_{RABS}V_Z^\dagger||V_Z\cZ(\sigma_{RABS})V_Z^\dagger)\\
	&= \inf_{\sigma_{RABS}}D_\gamma(\cG(\rho_{AB})||\cZ(\sigma_{RABS})).
\end{align}
The first inequality is taken from Eq. (45) of Ref \cite{zhu2017coherence}, where $E_{R,\gamma}^{Z_A|RABS}$ is the \Renyi relative entropy of entanglement between $Z_A$ and $RABS$. The second inequality is then obtained from the fact that $V_Z \cZ(\rho_{RABS})V_Z^\dagger$ is a separable state as noted by Ref \cite{coles2012unification}, since the system $R$ in $\cZ(\rho_{RABS})$ is decohered in $Z$ due to the projective measurement. As noted in Ref \cite{coles2012unification}, decoherence of a state occurs when a projective measurement is applied to a system's reduced state, resulting in the loss of off-diagonal elements, therefore becoming a separable state. The final equality comes from the isometric invariance of the \Renyi divergence as noted in Ref \cite{tomamichel2015quantum} ($V_Z^\dagger  V_Z = \II$). From the above, $\inf_{\sigma_{RABS}}D_\gamma(\cG(\rho_{AB})||\cZ(\sigma_{RABS}))\leq H_\alpha(Z_A|E')\leq \inf_{\sigma_{RABS}}D_\gamma(\cG(\rho_{AB})||\cZ(\sigma_{RABS}))$, therefore all the above inequalities saturate. Giving us
\begin{equation}
    H_\alpha(Z_A|E')=\inf_{\sigma_{RABS}}D_\gamma(\cG(\rho_{AB})||\cZ(\sigma_{RABS})),
    \label{eq: infD bound}
\end{equation}
thus establishing Eq.~\eqref{eq:D_gamma}. We next proceed to lower bound the \Renyi divergence. We first state the definition for the sandwiched (down-arrow) entropy $\tilde{H}^\downarrow_\alpha$ as
\begin{equation}
H_\alpha^\downarrow (A|B)_\rho := -D_\alpha(\rho_{AB}||\II_A \otimes \rho_B).
\end{equation}
We also have the following definition for Petz \Renyi entropy
\begin{align}
    \bar{H}_\alpha(A|B)_\rho &:= -\inf_{\sigma_B}\bar{D}_\alpha(\rho_{AB}||\II_A\otimes \sigma_{B})\\
    \bar{H}_\alpha^\downarrow(A|B)_\rho &:= -\bar{D}_\alpha(\rho_{AB}||\II_A\otimes\rho_B),
\end{align}
where $\bar{D}_\alpha(\rho||\sigma) = \frac{1}{\alpha-1}\log_2\frac{\Tr(\rho^\alpha \sigma^{1-\alpha})}{\Tr{\rho}}$  is the Petz \Renyi divergence \cite{tomamichel2015quantum}. Then, for $\beta = \alpha^{-1}$ where $\alpha\in(1,2]$, we have
\begin{align}
    H_\alpha(Z_A|E') 
    &\geq \bar{H}_\alpha(Z_A|E') \label{eq: sandwiched to petz}\\
    &= -H_\beta^\downarrow(Z_A|RABS) \label{eq: petz duality}\\
    &= D_\beta(\rho_{Z_ARABS}||\II_{Z_A}\otimes \rho_{RABS})\\
    &\geq D_\beta(V_Z^\dagger \rho_{Z_ARABS}V_Z||V_Z^\dagger(\II_{Z_A}\otimes\rho_{RABS})V_Z)\\
    &=D_\beta(\cG(\rho_{AB})||\cZ\circ\cG(\rho_{AB})).
\end{align}
Eq.~\eqref{eq: sandwiched to petz} is given by Eq. (
5.69) of Ref \cite{tomamichel2015quantum}, while Eq.~\eqref{eq: petz duality} is given by Proposition 5.8 of Ref \cite{tomamichel2015quantum}. The rest of the proof follows similarly to the proof for Eq.~\eqref{eq: infD bound}. Using Eq.~\eqref{eq: infD bound} thus gives us $\inf_{\sigma_{RABS}}D_\gamma(\cG(\rho_{AB})||\cZ(\sigma_{RABS})) \geq D_\beta(\cG(\rho_{AB})||\cZ\circ\cG(\rho_{AB}))$ which proves Eq.~\eqref{eq: divergence bound} and concludes Theorem \ref{thm: duality renyi divergence}.

\section{Proof of Theorem \ref{thm: gradient func}}\label{app:gradproof}
We begin this section by first deriving the gradient of the sandwiched \Renyi divergence as defined in Def.~\ref{def:entropy}. To begin, we first use the following shorthand $Q_\beta$:
\begin{equation}
    Q_\beta(\cG(\rho)||\cZ\circ\cG(\rho))\equiv\pnormp{\Xi}{\beta} = \Tr\left(\Xi\right)^\beta
\end{equation}
where 
\begin{equation}
    \Xi \equiv \cZ\circ\cG(\rho)^{\mu}\cdot \cG(\rho)\cdot \cZ\circ\cG(\rho)^{\mu}
\end{equation}
and we have denoted $\mu = \frac{1-\beta}{2\beta}$. The \Renyi divergence is related to $Q_\beta$ via $D_\beta(\rho||\sigma) \equiv \frac{1}{\beta-1}\log_2\left(\frac{Q_\beta(\rho||\sigma)}{\Tr(\rho)}\right)$. We next need to show the derivative of $D_\beta(\cG(\rho)||\cZ\circ\cG(\rho))$. To begin, we introduce the Fr{\'e}chet's derivative (Sec. V.3 and Sec. X.4 of Ref \cite{bhatia2013matrix}) of the \Renyi divergence in the direction of $\tau$ at a point $\rho$ as
\begin{align}
    \partial_\tau D_\beta(\cG(\rho) \| \cZ\circ\cG(\rho))
    &\left. \equiv \partial_x D_\beta\bigl(\cG(\rho_x) \| \cZ\circ\cG\left(\rho_x\right)\bigr)\right|_{x=0} \\
    & \equiv \Tr[(\tau-\rho)\nabla D_\beta(\cG(\rho)\|\cZ\circ\cG(\rho))],
    \label{eq: frechet Divergence}
\end{align}
\noindent where we have denoted $\rho_x \equiv (1-x)\rho+x\tau$ as the state $\rho$ being parameterized with a scalar $x$, with $\partial_\tau = \frac{\partial}{\partial \tau}$ and $\partial_x = \frac{\partial}{\partial x}$. Then, the derivative of $D_\beta(\cG(\rho_x)||\cZ\circ\cG(\rho_x))$ with respect to $x$ is
\begin{equation}
    \partial_x D_\beta(\cG(\rho_x)||\cZ\circ\cG(\rho_x))|_{x=0} = \frac{1}{\beta-1}\left.\frac{\Tr(\cG(\rho_x))}{Q_\beta(\cG(\rho_x)||\cZ\circ\cG(\rho_x))}\partial_x \left( \frac{Q_\beta (\cG(\rho_x)||\cZ\circ\cG(\rho_x))}{\Tr(\cG(\rho_x))}\right)\right|_{x=0}.
\label{eq: derivative of divergence wrt x}
\end{equation}
Applying the quotient rule allows us to rewrite Eq.~\eqref{eq: derivative of divergence wrt x} as
\begin{equation}
    \partial_x D_\beta(\cG(\rho_x)||\cZ\circ\cG(\rho_x))|_{x=0} = \frac{1}{\beta-1}\left[\frac{\partial_x Q_\beta(\cG(\rho_x)||\cZ\circ\cG(\rho_x))}{Q_\beta(\cG(\rho_x)||\cZ\circ\cG(\rho_x))} - \frac{\partial_x \Tr(\cG(\rho_x))}{\Tr(\cG(\rho_x))}\right]_{x=0}.
    \label{eq: derivative of divergence wrt x quotient rule}
\end{equation}
Next we derive the gradient of $Q_\beta$ via the lemma below.
\begin{lemma}
Consider a state $\rho$, we have that for $\beta\in(0,1)$ and $\mu = \frac{1-\beta}{2\beta}$, 
\begin{equation}
    \partial_\tau Q_\beta(\cG(\rho)||\cZ\circ\cG(\rho)) = \Tr[\cG^\dagger(\chi_1+\chi_2+\chi_3)(\tau-\rho)]
\end{equation}
where 
\begin{align}
    \chi_1 &= \cZ\left(\beta \frac{\sin(\pi\mu)}{\pi}\int_0^\infty (\cZ\circ\cG(\rho)+t)^{-1}\cG(\rho)\cZ\circ\cG(\rho)^\mu \Xi^{\beta-1}(\cZ\circ\cG(\rho)+t)^{-1}t^\mu dt\right)\\
    \chi_2 &= \beta \cZ\circ\cG(\rho)^\mu \Xi^{\beta-1}\cZ\circ\cG(\rho)^\mu\\
    \chi_3 &= \cZ\left(\beta \frac{\sin(\pi\mu)}{\pi}\int_0^\infty (\cZ\circ\cG(\rho)+t)^{-1}\Xi^{\beta-1}\cZ\circ\cG(\rho)^\mu\cG(\rho)(\cZ\circ\cG(\rho)+t)^{-1}t^\mu dt\right),
\end{align}
for $\cG^\dagger(\cdot) = \sum_s G_s^\dagger (\cdot) G_s$, and $\cZ(\cdot)=\sum_i(Z^i_R\otimes \II_{ABS})(\cdot)(Z^i_R\otimes \II_{ABS})$.
\label{lm: Q derivative}
\end{lemma}
\begin{proof}
First, by noting that derivative and trace commute, we have that
\begin{align}
\partial_\tau Q_\beta(\cG(\rho)||\cZ\circ\cG(\rho)) = \partial_x \Tr(\Xi_x^\beta)|_{x=0}= \Tr(\partial_x \Xi_x^\beta)|_{x=0},
\end{align}
where $\Xi_x \equiv \cZ\circ\cG(\rho_x)^{\mu}\cdot \cG(\rho_x)\cdot \cZ\circ\cG(\rho_x)^{\mu}$.
Next, from Lemma 1 of Ref \cite{rubboli2024new}, we see that
\begin{equation}
\partial_x (\Xi_x^\beta ) = \frac{\sin(\pi\beta)}{\pi}\int_0^\infty (\Xi_x+t)^{-1}(\partial_x \Xi_x)(\Xi_x+t)^{-1}t^\beta dt = \frac{\sin(\pi\beta)}{\pi}\int_0^\infty \frac{\partial_x \Xi_x}{(\Xi_x+t)^2}t^\beta dt.
\end{equation}
Taking the trace, we have
\begin{align}
    \Tr[\partial_x (\Xi_x^\beta)] &= \frac{\sin(\pi\beta)}{\pi} \int_0^\infty \Tr[(\Xi_x+t)^{-1}(\partial_x\Xi_x)(\Xi_x+t)^{-1}]t^\beta dt \label{eq: Xi diff 1} \\
    &= \frac{\sin(\pi\beta)}{\pi} \int_0^\infty \Tr[(\Xi_x+t)^{-1}(\Xi_x+t)^{-1}(\partial_x\Xi_x)]t^\beta dt  \label{eq: Xi diff 2}\\
    &= \frac{\sin(\pi\beta)}{\pi} \int_0^\infty \Tr\left[\frac{1}{(\Xi_x+t)^2}(\partial_x\Xi_x)\right]t^\beta dt \label{eq: Xi diff 3}\\
    &= \Tr[\beta\Xi_x^{\beta-1}(\partial_x\Xi_x)] \label{eq: Xi diff 4}.
\end{align}
From the above, Eq. \eqref{eq: Xi diff 2} is due to the cyclicity of the trace, whereas Eq. \eqref{eq: Xi diff 4} is due to the fact that 
\begin{equation}
    \int_0^\infty \frac{1}{(\Xi_x+t)^2}t^\beta dt = \frac{\pi\beta}{\sin(\pi\beta)}\Xi_x^{\beta-1}.
\end{equation}
Next, using the product rule $\partial (X\cdot Y\cdot Z) = (\partial X)\cdot Y\cdot Z+X\cdot(\partial Y)\cdot Z+X\cdot Y\cdot (\partial Z)$, we can then obtain the following derivative.
\begin{align}
    &\partial_\tau Q_\beta(\cG(\rho)||\cZ\circ\cG(\rho)) \\
    &= \left.\partial_x \Tr\left(\Xi_x^\beta\right)\right|_{x=0}\\
    &= \left.\Tr\left[\beta\Xi_x^{\beta-1}(\partial_x\Xi_x)\right]\right|_{x=0}\\
    &= \Tr\Bigl[\beta\Xi_x^{\beta-1}\bigl[\partial_x (\cZ\circ\cG(\rho_x)^\mu)\cG(\rho_x) \cZ\circ\cG(\rho_x)^{\mu}  + \cZ\circ\cG(\rho_x)^\mu \partial_x (\cG(\rho_x))\cZ\circ\cG(\rho_x)^\mu \label{eq: Xi derivative step} \\
    & \left.\quad +\cZ\circ\cG(\rho_x)^\mu\cG(\rho_x)\partial_x (\cZ\circ\cG(\rho_x)^\mu)\bigr]\Bigr]\right|_{x=0} \nonumber,
\end{align}
Again, from Lemma 1 of Ref \cite{rubboli2024new}, we see that
\begin{align}
    \partial_x (\cZ\circ\cG(\rho_x)^{\mu}) &= \frac{\sin(\pi\mu)}{\pi}\int_0^\infty (\cZ\circ\cG(\rho_x)+t)^{-1}\partial_x \left(\cZ\circ\cG(\rho_x)\right)(\cZ\circ\cG(\rho_x)+t)^{-1}t^\mu dt \\
    &= \frac{\sin(\pi \mu)}{\pi}\int_0^\infty \frac{\cZ\circ\cG(\tau-\rho)}{(\cZ\circ\cG(\rho_x) + t)^2} t^{\mu} dt.
\end{align}
The first term of Eq. \eqref{eq: Xi derivative step} are as follows.
\begin{align}
    &\Tr[\beta\Xi_x^{\beta-1}\partial_x(\cZ\circ\cG(\rho_x)^\mu) \cG(\rho_x)\cZ\circ\cG(\rho_x)^\mu]_{x=0}\\
    &=\Tr\left[\beta\Xi^{\beta-1}\frac{\sin(\pi\mu)}{\pi}\int_0^\infty (\cZ\circ\cG(\rho)+t)^{-1}\cZ\circ\cG(\tau-\rho)(\cZ\circ\cG(\rho)+t)^{-1} t^\mu dt \cdot \cG(\rho)\cZ\circ\cG(\rho)^\mu\right] \label{eq: sub1}\\
    &= \beta\frac{\sin(\pi\mu)}{\pi}\int_0^\infty\Tr[\Xi^{\beta-1}(\cZ\circ\cG(\rho)+t)^{-1}\cZ\circ\cG(\tau-\rho)(\cZ\circ\cG(\rho)+t)^{-1}\cG(\rho)\cZ\circ\cG(\rho)^\mu]t^\mu dt \label{eq: sub2}\\
    &= \Tr\Bigl[\cG^\dagger\circ\cZ\Bigl(\beta\frac{\sin(\pi\mu)}{\pi}\int_0^\infty (\cZ\circ\cG(\rho)+t)^{-1}\cG(\rho)\cZ\circ\cG(\rho)^\mu\Xi^{\beta-1}(\cZ\circ\cG(\rho)+t)^{-1}t^\mu dt\Bigr)(\tau-\rho)\Bigr] \label{eq: sub3}
\end{align}
Eq. \eqref{eq: sub2} is due to the fact that the integral and trace commute, while Eq. \eqref{eq: sub3} is obtained from the cyclicity of the trace. Next, the second term of Eq. \eqref{eq: Xi derivative step} is 
\begin{align}
    \Tr[\beta\Xi_x^{\beta-1}\cZ\circ\cG(\rho_x)^\mu\partial_x(\cG(\rho_x))\cZ\circ\cG(\rho_x)^\mu]_{x=0} &= \Tr[\beta\Xi^{\beta-1}\cZ\circ\cG(\rho)^\mu \cG(\tau-\rho)\cZ\circ\cG(\rho)^\mu]\\
    &= \Tr[\cG^\dagger \bigl(\beta\cZ\circ\cG(\rho)^\mu \Xi^{\beta-1}\cZ\circ\cG(\rho)^\mu\bigr)(\tau-\rho)]
\end{align}
Finally, the third term of Eq. \eqref{eq: Xi derivative step} is
\begin{align}
    &\Tr[\beta\Xi_x^{\beta-1}\cZ\circ\cG(\rho_x)^\mu \cG(\rho_x)\partial_x(\cZ\circ\cG(\rho_x)^\mu)]_{x=0}\\
    &=\Tr\left[\beta\Xi^{\beta-1}\frac{\sin(\pi\mu)}{\pi}\int_0^\infty \cZ\circ\cG(\rho)^\mu \cG(\rho)(\cZ\circ\cG(\rho)+t)^{-1}\cZ\circ\cG(\tau-\rho)(\cZ\circ\cG(\rho)+t)^{-1}t^\mu dt\right]\\
    &= \beta\frac{\sin(\pi\mu)}{\pi}\int_0^\infty\Tr\left[\Xi^{\beta-1}\cZ\circ\cG(\rho)^\mu \cG(\rho)(\cZ\circ\cG(\rho)+t)^{-1}\cZ\circ\cG(\tau-\rho)(\cZ\circ\cG(\rho)+t)^{-1}\right]t^\mu dt\\
    &= \Tr\Bigl[\cG^\dagger\circ\cZ\Bigl(\beta\frac{\sin(\pi\mu)}{\mu}\int_0^\infty(\cZ\circ\cG(\rho)+t)^{-1}\Xi^{\beta-1}\cZ\circ\cG(\rho)^\mu \cG(\rho)(\cZ\circ\cG(\rho)+t)^{-1}t^\mu dt\Bigr)(\tau-\rho)\Bigr]
\end{align}
This concludes the proof for Lemma \ref{lm: Q derivative}.
\end{proof}

\begin{theorem}\label{thm: gradient D} For states $\cG(\rho)$ and $\cZ\circ\cG(\rho)$, the gradient of the \Renyi divergence is given by
    \begin{equation}
    \nabla D_\beta\Bigl(\cG(\rho)\bigl|\bigr|\cZ\circ\cG(\rho)\Bigr) = 
        \frac{1}{(\beta-1)}\cG^\dagger\left[\frac{\chi_1 + \chi_2 + \chi_3}{Q_\beta (\cG(\rho)||\cZ\circ\cG(\rho))} - \frac{\II}{\Tr(\cG(\rho))}\right],
\end{equation}
where
\begin{align}
    \chi_1 &= \cZ\left(\beta \frac{\sin(\pi\mu)}{\pi}\int_0^\infty (\cZ\circ\cG(\rho)+t)^{-1}\cG(\rho)\cZ\circ\cG(\rho)^\mu \Xi^{\beta-1}(\cZ\circ\cG(\rho)+t)^{-1}t^\mu dt\right)\\
    \chi_2 &= \beta \cZ\circ\cG(\rho)^\mu \Xi^{\beta-1}\cZ\circ\cG(\rho)^\mu\\
    \chi_3 &= \cZ\left(\beta \frac{\sin(\pi\mu)}{\pi}\int_0^\infty (\cZ\circ\cG(\rho)+t)^{-1}\Xi^{\beta-1}\cZ\circ\cG(\rho)^\mu\cG(\rho)(\cZ\circ\cG(\rho)+t)^{-1}t^\mu dt\right).
\end{align}
\end{theorem}
\begin{proof} 
By substituting Lemma \ref{lm: Q derivative} into Eq.~\eqref{eq: derivative of divergence wrt x quotient rule}, and using the fact that $\partial_x \Tr(\cG(\rho_x)) = \Tr[(\tau-\rho)\cdot \cG^\dagger(\II)]$ gives
\begin{equation}
\begin{aligned}
    &\partial_x D_\beta\left(\cG(\rho_x) \| \cZ\circ\cG\left(\rho_x\right)\right)|_{x=0}, \\
    &= \Tr \left[(\tau-\rho)\cdot\frac{1}{\beta-1} \cG^\dagger\left(\frac{\chi_{1}+\chi_{2}+\chi_3}{Q_\beta\left(\cG(\rho) \| \cZ\circ\cG\left(\rho\right)\right)}-\frac{\II}{\Tr(\cG(\rho))}\right)\right]\\
    &\equiv \Tr\left[(\tau-\rho)\cdot\nabla D_\beta (\cG(\rho)\|\cZ\circ\cG(\rho_x))\right].
\end{aligned}
\end{equation}

This concludes our proof for Theorem \ref{thm: gradient D}.
\end{proof}
Finally, we are in the position to prove Theorem \ref{thm: gradient func}. Given that $f_\beta(\rho)=\Tr(\cG(\rho))\cdot D_\beta\bigl(\cG(\rho)\|\cZ\circ\cG(\rho)\bigr)$, 
using product rule and Theorem \ref{thm: gradient D}, we have
\begin{equation}
\begin{aligned}
    \partial_x f_\beta(\rho_x)|_{x=0} =& \left.\partial_x \Bigl(\Tr[\cG(\rho_x)]\Bigr)\cdot D_\beta\left(\cG(\rho_x) \| \cZ\circ\cG\left(\rho_x\right)\right) \right|_{x=0} + \Tr[\cG(\rho_x)] \cdot \left.\partial_x \Bigl(D_\beta\left(\cG(\rho_x) \| \cZ\circ\cG\left(\rho_x\right)\right)\Bigr) \right|_{x=0}\\
    =& \Tr[(\tau-\rho)\cG^\dagger(\II)]\cdot D_\beta(\cG(\rho)||\cZ\circ\cG(\rho)) + \Tr[\cG(\rho)]\cdot \Tr\bigl[(\tau-\rho)\cdot\nabla D_\beta(\cG(\rho)\|\cZ\circ\cg(\rho))\bigr].
\end{aligned}
\end{equation} 
We then use the fact that $D_\beta(\cG(\rho) \| \cZ\circ\cG\left(\rho\right))$ and $\Tr[\cG(\rho)]$ are scalars, hence 
\begin{align}
    \Tr\bigl[(\tau-\rho)\cG^\dagger(\II)\bigr]\cdot D_\beta(\cG(\rho) \| \cZ\circ\cG\left(\rho\right)) &= \Tr\Bigl\{(\tau-\rho)\cG^\dagger\bigl[\II\cdot D_\beta(\cG(\rho) \| \cZ\circ\cG\left(\rho\right))\bigr]\Bigr\} \\
    \Tr[\cG(\rho)]\cdot \Tr\bigl[(\tau-\rho)\cdot\nabla D_\beta(\cG(\rho)\| \cZ\circ\cG\left(\rho\right))\bigr] &= \Tr\bigl[(\tau-\rho)\Tr[\cG(\rho)]\cdot \nabla D_\beta(\cG(\rho) \| \cZ\circ\cG\left(\rho\right))\bigr].
\end{align}
This gives us the resulting relation
\begin{align}
    \partial_x f_\beta(\rho_x)|_{x=0} &= \Tr\biggl\{(\tau-\rho)\Bigl[\cG^\dagger\bigl[\II\cdot D_\beta(\cG(\rho) \| \cZ\circ\cG\left(\rho\right))\bigr]+ \Tr[\cG(\rho)] \cdot \nabla D_\beta (\cG(\rho) \| \cZ\circ\cG\left(\rho\right))\Bigr]\biggr\} \\
    &\equiv \Tr\bigl[(\tau-\rho)\nabla f_\beta(\rho)\bigr].
\end{align}
This concludes our proof for Theorem \ref{thm: gradient func}.

\vspace{-0.3cm}
\section{SDP Formulation with Finite-Size Constraints}\label{appendix:constraint}
In this section we formulate the SDP problem based on Eqs.~\eqref{eq: step 1 optimization} and ~\eqref{eq: step 1 set} based on Ref \cite{george2021numerical}. The primal problem is as such:
\begin{equation}
    \begin{aligned} 
    \min \quad & \langle\nabla f(\rho), \sigma\rangle \\ 
    \text {s.t.} \quad & \operatorname{Tr}\left(\Gamma_i \sigma\right)=\gamma_i \quad \forall i \in \Sigma, \\ & \operatorname{Tr}\left(A\right)+\operatorname{Tr}\left(B\right) \leqslant \mu, \\ 
    & A \succeq \Phi_{\mathcal{P}}(\sigma)-\bar{F}, \\ 
    & B \succeq-\left[\Phi_{\mathcal{P}}(\sigma)-\bar{F}\right]\\
    & \sigma, A,B \geq 0.
    \end{aligned}\label{eq: primal sdp}
\end{equation}
The corresponding dual problem is
\begin{equation}
    \begin{aligned} 
    \max \quad & \Vec{\gamma}\cdot \Vec{y}+\bar{f}\cdot \bar{z}-a\mu \\ 
    \text {s.t.} \quad & \sum_iy_i\Gamma_i + \sum_jz_j\tilde{\Gamma}_j\leq f(\rho) \\ & -a\II\leq \Vec{z}\leq a\II\\ 
    & a \geq 0, \Vec{y}\in \RR^{|\Sigma|},
    \end{aligned}\label{eq: dual sdp}
\end{equation}
where $\bar{f}$ is the vector version of $\bar{F}$.

\end{document}